\documentclass[
    11pt, 
]{article}
\usepackage[scale=.8]{geometry}

\date{}

\newcommand{\papertitle}{DPO: Dynamic-Programming Optimization on Hybrid Constraints}

\title{
    \papertitle%
    \footnote{Work supported in part by NSF grants IIS-1527668, CCF-1704883, IIS-1830549, and CNS-2016656; DoD MURI grant N00014-20-1-2787; and an award from the Maryland Procurement Office.}
}

\newcommand{\authors}{Vu H. N. Phan and Moshe Y. Vardi}

\author{
    \authors \\
    \texttt{\{vhp1,vardi\}@rice.edu} \\
    Rice University
}

\usepackage[
    pdftitle=\papertitle,
    pdfauthor=\authors,
    hidelinks,
]{hyperref}

\usepackage[commandnameprefix=ifneeded]{changes}
\definechangesauthor[name=Moshe, color=red]{myv}
\definechangesauthor[name=Vu, color=blue]{vhp}

\usepackage[none]{hyphenat}

\usepackage{breakcites}

\usepackage{amsthm}

\newtheorem{definition}{Definition}
\newtheorem{lemma}{Lemma}

\usepackage{thm-restate}

\usepackage{bookmark}
\bookmarksetup{
    numbered,
    open,
}

\usepackage{float}

\usepackage{multirow}

\usepackage{tikz-qtree}

\usepackage[
    ruled,
    linesnumbered,
]{algorithm2e}
\SetCommentSty{}
\SetKw{Assert}{assert}
\SetKwFor{For}{for}{}{}
\SetKwFor{While}{while}{}{}
\SetKwIF{If}{ElseIf}{Else}{if}{}{else if}{else}{}

\usepackage{amsmath}
\usepackage[
    capitalize,
    noabbrev,
]{cleveref}

\newcommand{\ms}{MaxSAT}

\newcommand{\bayescnfs}{1049}

\newcommand{\chaincnfs}{441}
\newcommand{\fastchaincnfs}{371}

\newcommand{\tool}[1]{\texttt{#1}}
\newcommand{\cash}{\tool{CASHWMaxSAT}}
\newcommand{\cms}{\tool{CryptoMiniSat}}

\newcommand{\dpmc}{\tool{DPMC}}

\newcommand{\dpo}{\tool{DPO}}

\newcommand{\executor}{\tool{DMC}}
\newcommand{\flowcutter}{\tool{FlowCutter}}
\newcommand{\gauss}{\tool{GaussMaxHS}}

\newcommand{\maxhs}{\tool{MaxHS}}
\newcommand{\mpesat}{\tool{MPE-SAT}}
\newcommand{\planner}{\tool{LG}}

\newcommand{\uwr}{\tool{UWrMaxSat}}
\newcommand{\vbs}[1]{\tool{VBS{#1}}}

\newcommand{\angles}[1]{\langle #1 \rangle}
\newcommand{\braces}[1]{\left\{ #1 \right\}}
\newcommand{\brackets}[1]{\left[ #1 \right]}
\newcommand{\pars}[1]{\left( #1 \right)}
\newcommand{\pipes}[1]{\left| #1 \right|}

\newcommand{\tup}{\angles} 

\newcommand{\set}{\braces}

\newcommand{\pb}[1]{\brackets{#1}} 

\newcommand{\of}{\pars}

\newcommand{\size}[1]{\pipes{#1}}

\newcommand{\bigo}[1]{\operatorname{O}\of{#1}}

\newcommand{\vars}[1]{\operatorname{vars}\of{#1}}

\newcommand{\dom}[1]{\operatorname{dom}\of{#1}} 

\newcommand{\width}[1]{\operatorname{width}\of{#1}}

\newcommand{\func}[1]{\mathtt{#1}}

\newcommand{\valuator}{\func{Valuator}}

\newcommand{\push}[2]{\func{push}\of{#1, #2}}
\newcommand{\pop}[1]{\func{pop}\of{#1}}

\newcommand{\minsert}[2]{\func{insert}\of{#1, #2}} 
\renewcommand{\minsert}[2]{#1 \gets #1 \cup \set{#2}}

\newcommand{\mremove}[2]{\func{remove}\of{#1, #2}} 
\renewcommand{\mremove}[2]{#1 \gets #1 \setminus \set{#2}}

\newcommand{\ta}{\tau} 

\newcommand{\va}[2]{\tup{#1, #2}} 

\newcommand{\extend}[3]{#1 \cup \set{\va{#2}{#3}}}

\newcommand{\wf}{W} 

\usepackage{stmaryrd}
\newcommand{\val}[1]{{\llbracket #1 \rrbracket}_\wf} 

\usepackage{amssymb}
\newcommand{\restrict}[2]{#1\restriction_{#2}}

\newcommand{\B}{\mathbb{B}}
\newcommand{\R}{\mathbb{R}}

\newcommand{\ps}[1]{\B^{#1}} 

\renewcommand{\emptyset}{\varnothing}

\renewcommand{\phi}{\varphi}

\newcommand{\V}[1]{\mathcal{V}\of{#1}} 

\newcommand{\C}[1]{\mathcal{C}\of{T, r, #1}} 

\newcommand{\Lv}{\mathcal{L}\of{T, r}} 

\newcommand{\T}{\mathcal{T}} 

\newcommand{\gammamap}{\overset{\gamma}{\mapsto}}
\newcommand{\pimap}{\overset{\pi}{\mapsto}}

\newcommand{\stack}{\sigma}

\DeclareMathOperator*{\dsgn}{dsgn}

\DeclareMathOperator*{\argmax}{argmax}

\usepackage{scalefnt} 
\newcommand{\scaleexists}[1]{\vcenter{\hbox{\scalefont{#1}$\exists$}}}
\DeclareMathOperator*{\bigexists}{\vphantom\sum\mathchoice{\scaleexists{2}}{\scaleexists{1.4}}{\scaleexists{1}}{\scaleexists{0.75}}}

\newcommand{\eg}{e.g.}
\newcommand{\ie}{i.e.}
\newcommand{\wrt}{w.r.t.}

\newcommand{\bmpe}{Boolean MPE}
\newcommand{\xcnf}{XOR-CNF}

\begin{document}


\maketitle


\begin{abstract}
    In Bayesian inference, the \emph{most probable explanation (MPE)} problem requests a variable instantiation with the highest probability given some evidence.
    Since a Bayesian network can be encoded as a literal-weighted CNF formula $\phi$, we study \emph{\bmpe}, a more general problem that requests a model $\tau$ of $\phi$ with the highest weight, where the weight of $\tau$ is the product of weights of literals satisfied by $\tau$.
    It is known that \bmpe{} can be solved via reduction to \emph{(weighted partial) \ms}.
    Recent work proposed \dpmc, a dynamic-programming model counter that leverages graph-decomposition techniques to construct \emph{project-join trees}.
    A project-join tree is an execution plan that specifies how to conjoin clauses and project out variables.
    We build on \dpmc{} and introduce \dpo, a dynamic-programming optimizer that exactly solves \bmpe.
    By using \emph{algebraic decision diagrams (ADDs)} to represent pseudo-Boolean (PB) functions, \dpo{} is able to handle disjunctive clauses as well as XOR clauses.
    (Cardinality constraints and PB constraints may also be compactly represented by ADDs, so one can further extend \dpo's support for hybrid inputs.)
    To test the competitiveness of \dpo, we generate random \xcnf{} formulas.
    On these hybrid benchmarks, \dpo{} significantly outperforms \maxhs, \uwr, and \gauss, which are state-of-the-art exact solvers for \ms.
\end{abstract}



\section{Introduction}

Bayesian inference \cite{pearl1985bayesian} has numerous applications, including medical diagnosis \cite{shwe1991probabilistic} and industrial fault analysis \cite{cai2017bayesian}.
Given a Bayesian network, the \emph{most probable explanation (MPE)} problem requests a variable instantiation with the highest probability.
MPE has the form $\argmax_X f\of{X}$, where $f$ is a real-valued function over a set $X$ of variables.
Also having that form is the \emph{satisfiability (SAT)} problem, which requests a \emph{satisfying truth assignment (model)} of a Boolean formula \cite{cook1971complexity}.
SAT can be viewed as a counterpart of MPE where $f$ is a Boolean function \cite{littman2001stochastic}.

A generalization of SAT is the \emph{\bmpe} problem, which receives a literal-weighted Boolean formula and requests a \emph{maximizing truth assignment (maximizer)}, \ie, a truth assignment $\ta$ with the highest weight, where the weight of $\ta$ is the product of weights of literals satisfied by $\ta$ \cite{sang2007dynamic}.
\bmpe{} is more general than (Bayesian) MPE because Bayesian networks can be encoded as literal-weighted Boolean formulas \cite{sang2005performing}.

A Boolean formula is usually given in \emph{conjunctive normal form (CNF)}, \ie, as a set of disjunctive clauses.
An optimization version of SAT, the \emph{maximum satisfiability (\ms)} problem, requests a truth assignment that maximizes the number of satisfied clauses of an unsatisfiable CNF formula \cite{krentel1988complexity}.
\ms{} may be used to solve MPE \cite{park2002using} and \bmpe{} \cite{sang2007dynamic} via appropriate reductions.

Most SAT and \ms{} solvers only accept Boolean formulas in CNF, but disjunctive clauses do not always provide a natural formulation of an application.
For example, certain properties in cryptography are more conveniently described with XOR \cite{bogdanov2011biclique}.
In addition, \emph{cardinality constraints} help formulate problems in graph coloring \cite{costa2009graph}, and \emph{hybrid constraints} have also been used for maximum-likelihood decoding \cite{feldman2005using}.

All Boolean constraints can be converted into CNF, and efficient CNF encodings often exist.
For example, XOR can be encoded using the Tseitin transformation, which adds a linear number of disjunctive clauses and auxiliary variables \cite{tseitin1983complexity}.
There are, however, many possible CNF encodings, and solver performance crucially depends on the encoding used, so finding good encodings is solver-dependent and nontrivial \cite{prestwich2009cnf}.
Thus, CNF conversion is not always a good approach when dealing with hybrid constraints.

Instead of using CNF encodings, one may wish to handle non-disjunctive constraints natively.
This direct approach avoids conversion overhead and preserves application-specific structure that would be hard to recover after transformation.
Hybrid solvers have been developed for SAT \cite{soos2009extending,yang2021engineering} and related problems, such as \emph{weighted model counting (WMC)} \cite{soos2019bird,soos2020tinted}.
(WMC requests the sum of weights of satisfying truth assignments of a Boolean formula \cite{valiant1979complexity}.)
Yet developing performant hybrid solvers for specific types of non-disjunctive constraints is labor-intensive, which motivates the need for hybrid solvers that can natively handle a variety of constraint types, \eg, \cite{kyrillidis2020fouriersat}.

A recent WMC framework is \dpmc{} \cite{dudek2020dpmc}, which employs \emph{dynamic programming}, a strategy that solves a large instance by decomposing it into parts then combining partial solutions into a final answer.
Dynamic programming has been applied to WMC \cite{fichte2020exploiting,dudek2021procount}, SAT \cite{pan2005symbolic}, \ms{} \cite{sang2007dynamic,saether2015solving}, and \emph{quantified Boolean formula (QBF)} evaluation \cite{charwat2016bdd}.

To guide dynamic programming for WMC, \dpmc{} uses \emph{project-join trees} as execution plans.
A project-join tree for a CNF formula specifies how to conjoin clauses and project out variables in order to obtain the weighted model count.
\dpmc{} operates in two phases.
First, the \emph{planning phase} builds a project-join tree with graph-decomposition techniques \cite{robertson1991graph,strasser2017computing}.
Second, the \emph{execution phase} uses the built tree to compute a final answer, where intermediate results are represented by \emph{algebraic decision diagrams (ADDs)} \cite{bahar1997algebraic,somenzi2015cudd}.
ADDs are a data structure that compactly represents pseudo-Boolean functions, such as conjunction, disjunction, XOR, cardinality constraints, and others.

We build on \dpmc{} to develop \dpo, a dynamic-programming optimization framework that exactly solves \bmpe{} on \xcnf{} formulas.
To find maximizers, we also adapt an iterative technique \cite{kyrillidis2022dpms} that was recently proposed for \ms.
To demonstrate the advantage of our approach, we generate random \xcnf{} benchmarks on which \dpo{} significantly outperforms \maxhs{} \cite{davies2011solving}, \uwr{} \cite{piotrow2020uwrmaxsat}, and \gauss{} \cite{soos2021gaussian}, which are state-of-the-art exact solvers for \ms.


\section{Related Work}
\label{secRelatedWork}

\mpesat{} \cite{sang2007dynamic} is an exact \bmpe{} solver that is based on DPLL and techniques in WMC \cite{sang2004combining}.
\mpesat{} dynamically detects \emph{connected components}, which are subformulas with disjoint sets of variables.
These components can be solved independently before results are combined into a final answer.
Processed components and their computed values are cached to avoid repeated work.
A lower bound is dynamically updated to prune the search tree.
\mpesat{} also utilizes CDCL to eliminate the infeasible search space.

Instead of solving \bmpe{} directly, one can employ a simple reduction from a \bmpe{} instance $\phi$ to a \emph{(weighted partial) \ms} instance $\psi$ \cite{sang2007dynamic}.
First, each clause of the original CNF formula $\phi$ induces a hard clause of the new CNF formula $\psi$.
Then, each literal $l$ of $\phi$ with weight $w$ induces a unit clause $l$ of $\psi$ with soft weight $\log\of{w}$.
If $\phi$ has $n$ variables and $m$ clauses, then $\psi$ has $n$ variables, $m$ hard clauses, and $2n$ soft clauses.
Notice that a truth assignment $\ta$ is a maximizer of $\phi$ if and only if $\ta$ is an optimal truth assignment of $\psi$.

There are several \ms{} algorithms that cover a variety of techniques \cite{bacchus2020maxsat}.
For example, \maxhs{} combines SAT and hitting-set computation \cite{davies2011solving}.
In more detail, \maxhs{} invokes a SAT solver to find \emph{unsatisfiable cores}.
To compute a minimal-cost hitting set, \maxhs{} employs \emph{integer linear programming (ILP)} as well as a branch-and-bound approach.
A more recent \ms{} solver is \uwr{} \cite{piotrow2020uwrmaxsat}, which also solves PB problems.
\uwr{} uses a core-guided approach and translates cardinality constraints into CNF.
For weighted instances, \uwr{} applies other techniques as well, such as preprocessing to detect unit cores.
An extension of \uwr{} is \cash{} \cite{lei2021cashwmaxsat}, which transforms small \ms{} instances into ILP.
Also, \cash{} invokes a SAT solver multiple times to find smaller cores.
After choosing a core, \cash{} delays encoding the new constraint if the core is larger than a predefined threshold.

Most \ms{} solvers only work on CNF, but \gauss{} \cite{soos2021gaussian} also supports XOR clauses.
Inspired by the SAT solver \cms{} \cite{soos2009extending}, \gauss{} uses an architecture that combines Gaussian elimination with the \ms{} solver \maxhs.
Gaussian elimination is performed on dense bit-packed matrices to exploit SIMD.
A related \xcnf{} architecture has been proposed for approximate WMC \cite{soos2019bird}.


\section{Preliminaries}

\subsection{Graphs}

In a \emph{graph} $G$, let $\V{G}$ denote the set of vertices.
A \emph{tree} is an undirected graph that is connected and acyclic.
We refer to a vertex of a tree as a \emph{node}.
A \emph{rooted tree} is a tree $T$ together with a distinguished node $r \in \V{T}$ called the \emph{root}.
In a rooted tree $\tup{T, r}$, each node $v \in \V{T}$ has a (possibly empty) set of \emph{children}, denoted by $\C{v}$, which contains every node $v'$ adjacent to $v$ such that the path from $v'$ to $r$ passes through $v$.
A \emph{leaf} of a rooted tree $\tup{T, r}$ is a non-root node of degree one.
Let $\Lv$ denote the set of leaves of $\tup{T, r}$.
An \emph{internal node} is a member of $\V{T} \setminus \Lv$, including the root.

\subsection{Pseudo-Boolean Functions}

Let $B^A$ denote the set of all functions $f$ with domain $\dom{f} := A$ and codomain $B$.
The \emph{restriction} of $f$ to a set $S$ is a function defined by $\restrict{f}{S} := \set{\tup{a, b} \in f \mid a \in S}$.

From now on, every variable is binary unless noted otherwise.
A \emph{truth assignment} for a set $X$ of variables is a function $\ta : X \to \B$.

A \emph{pseudo-Boolean (PB) function} over a set $X$ of variables is a function $f : \ps{X} \to \R$.
Define $\vars{f} := X$.
We say that $f$ is \emph{constant} if $X = \emptyset$.
A \emph{Boolean function} is a special PB function $f : \ps{X} \to \B$.

\begin{definition}[Join]
\label{defJoin}
    Let $f : \ps{X} \to \R$ and $g : \ps{Y} \to \R$ be PB functions.
    The \emph{(multiplicative) join} of $f$ and $g$ is a PB function, denoted by $f \cdot g : \ps{X \cup Y} \to \R$, defined for each $\ta \in \ps{X \cup Y}$ by $(f \cdot g)(\ta) := f(\restrict{\ta}{X}) \cdot g(\restrict{\ta}{Y})$.
\end{definition}

Join is commutative and associative: we have $f \cdot g = g \cdot f$ as well as $(f \cdot g) \cdot h = f \cdot (g \cdot h)$ for all PB functions $f$, $g$, and $h$.
Then define $\prod_{i = 1}^n f_i := f_1 \cdot f_2 \cdot \ldots \cdot f_n$.

\begin{definition}[Existential Projection]
\label{defExistProj}
    Let $f : \ps{X} \to \R$ be a PB function and $x$ be a variable.
    The \emph{existential projection} of $f$ \wrt{} $x$ is a PB function, denoted by $\exists_x f : \ps{X \setminus \set{x}} \to \R$, defined for each $\ta \in \ps{X \setminus \set{x}}$ by $(\exists_x f)(\ta) := \max\of{f(\extend{\ta}{x}{0}), f(\extend{\ta}{x}{1})}$.
\end{definition}

Existential projection is commutative: $\exists_x (\exists_y f) = \exists_y (\exists_x f)$ for all variables $x$ and $y$.
Then define $\exists_S f := \exists_x \exists_y \ldots f$, where $S = \set{x, y, \ldots}$ is a set of variables.
By convention, $\exists_\emptyset f := f$.

\begin{definition}[Maximum]
\label{defMaximum}
    Let $f : \ps{X} \to \R$ be a PB function.
    The \emph{maximum} of $f$ is the real number $(\exists_X f)(\emptyset)$.
\end{definition}

\begin{definition}[Maximizer]
\label{defMaximizer}
    Let $f : \ps{X} \to \R$ be a PB function.
    A \emph{maximizer} of $f$ is a truth assignment $\ta \in \ps{X}$ such that $f(\ta) = (\exists_X f)(\emptyset)$.
\end{definition}

Another type of projection, also widely used in propositional logic, is defined as follows.

\begin{definition}[Additive Projection]
\label{defAddProj}
    Let $f : \ps{X} \to \R$ be a PB function and $x$ be a variable.
    The \emph{additive projection} of $f$ \wrt{} $x$ is a PB function, denoted by $\Sigma_x f : \ps{X \setminus \set{x}} \to \R$, defined for each $\ta \in \ps{X \setminus \set{x}}$ by $(\Sigma_x f)(\ta) := f(\extend{\ta}{x}{0}) + f(\extend{\ta}{x}{1})$.
\end{definition}

Additive projection is commutative: $\Sigma_x (\Sigma_y f) = \Sigma_y (\Sigma_x f)$ for all variables $x$ and $y$.
Then define $\Sigma_S f := \Sigma_x \Sigma_y \ldots f$, where $S = \set{x, y, \ldots}$ is a set of variables.
By convention, $\Sigma_\emptyset f := f$.

Propositional logic can be generalized for probabilistic domains using real-valued weights.

\begin{definition}[Literal-Weight Function]
    A \emph{literal-weight function} over a set $X$ of variables is a PB function, denoted by $\wf : \ps{X} \to \R$, defined by $\wf := \prod_{x \in X} \wf_x$ for some PB functions $\wf_x : \ps{\set{x}} \to \R$.
\end{definition}

\subsection{\bmpe}

Given a Boolean formula $\phi$, define $\vars{\phi}$ to be the set of all variables that appear in $\phi$.
Then $\phi$ represents a Boolean function, denoted by $\pb{\phi} : \ps{\vars{\phi}} \to \B$, defined according to standard Boolean semantics.

\begin{definition}[\bmpe]
\label{defBmpe}
    Let $\phi$ be a Boolean formula and $\wf$ be a literal-weight function over $\vars{\phi}$.
    The \emph{\bmpe} problem on $\tup{\phi, \wf}$ requests the maximum and a maximizer of $\pb{\phi} \cdot \wf$.
\end{definition}

We also formally define the following related problem.

\begin{definition}[Weighted Model Counting]
    Let $\phi$ be a Boolean formula and $\wf$ be a literal-weight function over $X = \vars{\phi}$.
    The \emph{weighted model counting (WMC)} problem on $\tup{\phi, \wf}$ requests the real number $\pars{\Sigma_X \pars{\pb{\phi} \cdot \wf}}(\emptyset)$.
\end{definition}


\section{Solving \bmpe}

\subsection{Monolithic Approach}

A Boolean formula is usually given in \emph{conjunctive normal form (CNF)}, \ie, as a set of \emph{clauses}.
We work with a more general format, \emph{\xcnf}, in which a clause is an XOR or a disjunction of literals.

Given an \xcnf{} formula $\phi$, we have the factorization $\pb{\phi} = \prod_{c \in \phi} \pb{c}$.
In this section, we present an inefficient algorithm to solve \bmpe{} that treats $\phi$ as a monolithic structure and ignores the \xcnf{} factored representation.
The next section describes a more efficient algorithm.

To find maximizers for \bmpe, we leverage the following idea, which originated from the \emph{basic algorithm} for PB programming \cite{crama1990basic} and was adapted for \ms{} \cite{kyrillidis2022dpms}.

\begin{definition}[Derivative Sign]
\label{defDsgn}
    Let $f : \ps{X} \to \R$ be a PB function and $x$ be a variable.
    The \emph{derivative sign} of $f$ \wrt{} $x$ is a function, denoted by $\dsgn_x f : \ps{X \setminus \set{x}} \to \B^{\set{x}}$, defined for each $\ta \in \ps{X \setminus \set{x}}$ by $\pars{\dsgn_x f}(\ta) := \set{\va{x}{1}}$ if $f(\extend{\ta}{x}{1}) \ge f(\extend{\ta}{x}{0})$, and $\pars{\dsgn_x f}(\ta) := \set{\va{x}{0}}$ otherwise.
\end{definition}

The following result leads to an iterative process to find maximizers of PB functions \cite{kyrillidis2022dpms}.

\begin{restatable}[Iterative Maximization]{proposition}{rePropIterMax}
\label{propIterMax}
    Let $f : \ps{X} \to \R$ be a PB function and $x$ be a variable.
    Assume that a truth assignment $\ta$ is a maximizer of $\exists_x f : \ps{X \setminus \set{x}} \to \R$.
    Then the truth assignment $\ta \cup \pars{\dsgn_x f}(\ta)$ is a maximizer of $f$.
\end{restatable}

\cref{monoAlgo} can be used to find maximizers \cite{kyrillidis2022dpms}.

\begin{algorithm}[H]
\caption{Computing the maximum and a maximizer of a PB function}
\label{monoAlgo}
    \KwIn{$f_n$: a PB function over a set $X_n = \set{x_1, x_2, \ldots, x_n}$ of variables}
    \KwOut{$m \in \R$: the maximum of $f_n$}
    \KwOut{$\ta_n \in \ps{X_n}$: a maximizer of $f_n$}

    \DontPrintSemicolon
    \For{$i = n, n - 1, \ldots, 2, 1$}{
        $f_{i - 1} \gets \exists_{x_i} f_i$ \tcp{$f_i$ is a PB function over $X_i = \set{x_1, x_2, \ldots, x_i}$}
    }
    $m \gets f_0(\emptyset)$ \tcp{$f_0 = \exists_{X_n} f_n$ is a constant PB function over $X_0 = \emptyset$} \label{lineMonoMaximum}
    $\ta_0 \gets \emptyset$ \tcp{$\ta_0$ is a maximizer of $f_0$ (and also the only input to $f_0$)} \label{lineMonoTa0}
    \For{$i = 1, 2, \ldots, n - 1, n$}{
        $\ta_i \gets \ta_{i - 1} \cup \pars{\dsgn_{x_i} f_i}(\ta_{i - 1})$ \tcp{$\ta_i$ is a maximizer of $f_i$ since $\ta_{i - 1}$ is a maximizer of $f_{i - 1}$ (by \cref{propIterMax})} \label{lineMonoTaN}
    }
    \Return $\tup{m, \ta_n}$
\end{algorithm}

\cref{monoAlgo} can be used to find maximizers \cite{kyrillidis2022dpms}.

\begin{restatable}[Correctness of \cref{monoAlgo}]{proposition}{rePropMono}
\label{propMonoAlgo}
    Let $\phi$ be a Boolean formula and $\wf$ be a literal-weight function over $\vars{\phi}$.
    \cref{monoAlgo} solves \bmpe{} on $\tup{\phi, \wf}$ given the input $f_n = \pb{\phi} \cdot \wf$.
\end{restatable}

\bmpe{} on $\tup{\phi, \wf}$ can be solved by calling \cref{monoAlgo} with $\pb{\phi} \cdot \wf$ as an input.
But the PB function $\pb{\phi} \cdot \wf$ may be too large to fit in main memory, making the computation slow or even impossible.
In the next section, we exploit the \xcnf{} factorization of $\phi$ and propose a more efficient solution.

\subsection{Dynamic Programming}

\bmpe{} on $\tup{\phi, \wf}$ involves the PB function $\pb{\phi} \cdot \wf = \prod_{c \in \phi} \pb{c} \cdot \prod_{x \in X} \wf_x$, where $\phi$ is an \xcnf{} formula and $\wf$ is a literal-weight function over $X = \vars{\phi}$.
Instead of projecting all variables in $X$ after joining all clauses, we can be more efficient and project some variables early as follows \cite{dudek2021procount}.

\begin{restatable}[Early Projection]{proposition}{rePropEarlyProj}
\label{propEarlyProj}
    Let $f : \ps{X} \to \R$ and $g : \ps{Y} \to \R$ be PB functions.
    Then for all $S \subseteq X \setminus Y$, we have $\exists_S (f \cdot g) = (\exists_S f) \cdot g$ and $\Sigma_S (f \cdot g) = (\Sigma_S f) \cdot g$.
\end{restatable}

Early projection can lead to smaller intermediate PB functions.
For example, the bottleneck in computing $\exists_S (f \cdot g)$ is $f \cdot g$ with size $s = \size{\vars{f \cdot g}} = \size{X \cup Y}$.
The bottleneck in computing $(\exists_S f) \cdot g$ is $f$ with size $s_1 = \size{\vars{f}} = \size{X}$ or is $(\exists_S f) \cdot g$ with size $s_2 = \size{\vars{(\exists_S f) \cdot g}} = \size{X \cup Y \setminus S}$.
Notice that $s \ge \max(s_1, s_2)$.
The difference is consequential since an operation on a PB function $h$ may take $\bigo{2^{\size{\vars{h}}}}$ time and space.

We can apply early projection systematically by adapting the following framework, \dpmc, which uses dynamic programming for WMC on CNF formulas \cite{dudek2020dpmc}.

\begin{definition}[Project-Join Tree]
\label{defPjt}
    Let $\phi$ be an \xcnf{} formula (\ie, a set of XOR clauses and disjunctive clauses).
    A \emph{project-join tree} for $\phi$ is a tuple $\T = \tup{T, r, \gamma, \pi}$, where:
    \begin{itemize}
        \item $\tup{T, r}$ is a rooted tree,
        \item $\gamma : \Lv \to \phi$ is a bijection, and
        \item $\pi : \V{T} \setminus \Lv \to \ps{\vars{\phi}}$ is a function.
    \end{itemize}
    A project-join must satisfy the following criteria:
    \begin{enumerate}
        \item The set $\set{\pi(v) \mid v \in \V{T} \setminus \Lv}$ is a partition of $\vars{\phi}$, where some $\pi(v)$ sets may be empty.
        \item For each internal node $v$, variable $x \in \pi(v)$, and clause $c \in \phi$, if $x \in \vars{c}$, then the leaf $\gamma^{-1}\of{c}$ is a descendant of $v$ in $\tup{T, r}$.
    \end{enumerate}
\end{definition}

For a leaf $v \in \Lv$, define $\vars{v} := \vars{\gamma(v)}$, \ie, the set of variables that appear in the clause $\gamma(v) \in \phi$.
For an internal node $v \in \V{T} \setminus \Lv$, define $\vars{v} := \pars{\bigcup_{v' \in \C{v}} \vars{v'}} \setminus \pi(v)$.

\cref{figPjt} illustrates a project-join tree.

\begin{figure}[H]
    \centering
    \begin{tikzpicture}[grow = down]
        \tikzset{level distance = 40pt, sibling distance = 20pt}
        \Tree [ .$n_{10}\pimap\emptyset$
            [ .$n_8\pimap\set{x_1}$
                [ .$n_6\pimap\set{x_2, x_4}$
                    [ .$n_1\gammamap{x_2 \oplus \neg x_4}$ ]
                ]
                [ .$n_7\pimap\set{x_6}$
                    [ .$n_2\gammamap{x_1 \vee x_6}$ ]
                ]
                [ .$n_3\gammamap{x_1}$ ]
            ]
            [ .$n_9\pimap\set{x_3, x_5}$
                [ .$n_4\gammamap{x_3 \oplus x_5}$ ]
                [ .$n_5\gammamap{\neg x_3 \vee \neg x_5}$ ]
            ]
        ]
    \end{tikzpicture}
    \caption{
        A project-join tree $\T = \tup{T, n_{10}, \gamma, \pi}$ for an \xcnf{} formula $\phi$.
        Each leaf corresponds to a clause of $\phi$ under $\gamma$.
        Each internal node corresponds to a set of variables of $\phi$ under $\pi$.
    }
\label{figPjt}
\end{figure}
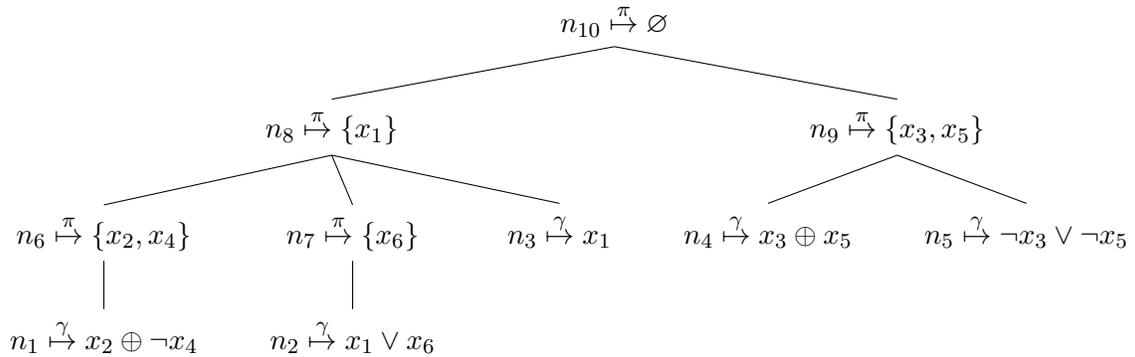

We adapt the following definition for \bmpe{} on \xcnf{} \cite{dudek2020dpmc}.
\begin{definition}[Valuation of Project-Join Tree Node]
\label{defValuation}
    Let $\phi$ be an \xcnf{} formula, $\tup{T, r, \gamma, \pi}$ be a project-join tree for $\phi$, and $\wf$ be a literal-weight function over $\vars{\phi}$.
    The \emph{$\wf$-valuation} of $v \in \V{T}$ is a PB function, denoted by $\val{v} : \ps{\vars{v}} \to \R$, defined by the following:
    \begin{align*}
        \val{v} :=
        \begin{cases}
            \pb{\gamma(v)} & \text{if } v \in \Lv \\
            \displaystyle
            \bigexists_{\pi(v)} \pars{ \prod_{v' \in \C{v}} \val{v'} \cdot \prod_{x \in \pi(n)} \wf_x } & \text{otherwise}
        \end{cases}
    \end{align*}
\end{definition}
Recall that $\pb{\gamma(v)}$ is the Boolean function represented by the clause $\gamma(v) \in \phi$.
Also, $\val{v'}$ is the $\wf$-valuation of a child $v'$ of the node $v$ in the rooted tree $\tup{T, r}$.
Note that the $\wf$-valuation of the root $r$ is a constant PB function.

\begin{definition}[Width of Project-Join Tree]
    Let $\T = \tup{T, r, \gamma, \pi}$ be a project-join tree.
    For a leaf $v$ of $\T$, define $\width{v} := \size{\vars{v}}$.
    For an internal node $v$ of $\T$, define $\width{v} := \size{\vars{v} \cup \pi(v)}$.
    The \emph{width} of $\T$ is $\width{\T} := \max_{v \in \V{T}} \width{v}$.
\end{definition}
Note that $\width{\T}$ is the maximum number of variables needed to valuate a node $v \in \V{T}$.
Valuating $\T$ may take $\bigo{2^{\width{\T}}}$ time and space.

We adapt the following correctness result \cite[Theorem 2]{dudek2020dpmc}.
\begin{restatable}[Valuation of Project-Join Tree Root]{theorem}{reThmRootVal}
\label{thmRootValuation}
    Let $\phi$ be a CNF formula, $\tup{T, r, \gamma, \pi}$ be a project-join tree for $\phi$, and $\wf$ be a literal-weight function over $X = \vars{\phi}$.
    Then $\exists_X \pars{\pb{\phi} \cdot \wf}(\emptyset) = \val{r}(\emptyset)$.
\end{restatable}
In other words, the valuation of the root $r$ is a constant PB function that maps $\emptyset$ to the maximum of $\pb{\phi} \cdot \wf$.

We now introduce \cref{dpAlgo}, which is more efficient than \cref{monoAlgo} due to the use of a project-join tree to systematically apply early projection.

\clearpage

\newcommand{\gx}{d} 

\begin{algorithm}[H]
\caption{Dynamic Programming for \bmpe}
\label{dpAlgo}
    \KwIn{$\phi$: an \xcnf{} formula}
    \KwIn{$\wf$: a literal-weight function over $X = \vars{\phi}$}
    \KwOut{$m \in \R$: the maximum of $\pb{\phi} \cdot \wf$}
    \KwOut{$\ta \in \ps{X}$: a maximizer of $\pb{\phi} \cdot \wf$}

    \DontPrintSemicolon
    $\T = \tup{T, r, \gamma, \pi} \gets$ a project-join tree for $\phi$\;
    $\stack \gets \tup{}$ \tcp{an initially empty stack}
    $\val{r} \gets \valuator(\phi, \T, \wf, r, \stack)$ \tcp{$\valuator$ (\cref{valuatorAlgo}) pushes derivative signs onto $\stack$}
    $m \gets \val{r}(\emptyset)$ \tcp{$\val{r}$ is a constant PB function}
    $\ta \gets \emptyset$ \tcp{an initially empty truth assignment}
    \While{$\stack$ is not empty}{
        $\gx \gets \pop{\stack}$ \tcp{$\gx = \dsgn_x f$ is a derivative sign, where $x$ is an unassigned variable and $f$ is a PB function}
        $\ta \gets \ta \cup \gx\of{\restrict{\ta}{\dom{\gx}}}$ \tcp{$\gx\of{\restrict{\ta}{\dom{\gx}}} = \set{\va{x}{b}}$, where $b \in \B$}
    }
    \Return $\tup{m, \ta}$ \tcp{all variables have been assigned in $\ta$}
\end{algorithm}

\begin{algorithm}[H]
\caption{$\valuator(\phi, \T, \wf, v, \stack)$}
\label{valuatorAlgo}
    \KwIn{$\phi$: an \xcnf{} formula}
    \KwIn{$\T = \tup{T, r, \gamma, \pi}$: a project-join tree for $\phi$}
    \KwIn{$\wf$: a literal-weight function over $\vars{\phi}$}
    \KwIn{$v \in \V{T}$: a node}
    \KwIn{$\stack$: a stack (of derivative signs) that will be modified}
    \KwOut{$\val{v}$: the $\wf$-valuation of $v$}

    \DontPrintSemicolon
    \If{$v \in \Lv$}{
        $f \gets \pb{\gamma(v)}$ \tcp{the Boolean function represented by the clause $\gamma(v) \in \phi$}
    }
    \Else(\tcp*[h]{$v$ is an internal node of $\tup{T, r}$}){
        $f \gets \prod_{v' \in \C{v}} \valuator(\phi, \T, \wf, v', \stack)$\;
        \For{$x \in \pi(v)$}{
            $\push{\stack}{\dsgn_x f}$ \tcp{$\stack$ is used to construct a maximizer (\cref{dpAlgo})}
            $f \gets \exists_x \pars{f \cdot \wf_x}$
        }
    }
    \Return $f$
\end{algorithm}

\begin{lemma}[Correctness of \cref{valuatorAlgo}]
\label{lemmaValuator}
    \cref{valuatorAlgo} returns the $\wf$-valuation of the input project-join tree node.
\end{lemma}
\begin{proof}
    \cref{valuatorAlgo} implements \cref{defValuation}.
    Modifying the input stack $\stack$ does not affect how the output valuation $f$ is computed.
\end{proof}

\begin{restatable}[Correctness of \cref{dpAlgo}]{theorem}{reThmDp}
\label{thmDpAlgo}
    Let $\phi$ be an \xcnf{} formula, $\tup{T, r, \gamma, \pi}$ be a project-join tree for $\phi$, and $\wf$ be a literal-weight function over $\vars{\phi}$.
    Then \cref{dpAlgo} solves \bmpe{} on $\tup{\phi, \wf}$.
\end{restatable}
\begin{proof}
    See \cref{secProofDp}.
\end{proof}

\cref{dpAlgo} comprises two phases: a \emph{planning phase} that builds a project-join tree $\T$ and an \emph{execution phase} that valuates $\T$.
The weight function $\wf$ is needed only in the execution phase.

We implemented \cref{dpAlgo} as \dpo, a dynamic-programming optimizer.
\dpo{} uses \planner{} \cite{dudek2020dpmc}, a planning tool that invokes \flowcutter, a solver \cite{strasser2017computing} for \emph{tree decompositions} \cite{robertson1991graph}.
A tree decomposition of a graph $G$ is a tree $T$, where each node of $T$ corresponds to a set of vertices of $G$ (plus other technical criteria).

Also, \dpo{} extends \executor{} \cite{dudek2020dpmc}, an execution tool that manipulates PB functions using \emph{algebraic decision diagrams (ADDs)} \cite{bahar1997algebraic}.
An ADD is a directed acyclic graph that can compactly represent a PB function and support some polynomial-time operations.
ADDs generalize \emph{binary decision diagrams (BDDs)} \cite{bryant1986graph}, which are used to manipulate Boolean functions.
ADDs and BDDs are implemented in the \tool{CUDD} library \cite{somenzi2015cudd}.

\cref{figAdd} illustrates an ADD.

\begin{figure}[H]
    \centering
    \includegraphics[
        height = 200pt,
        trim = {140pt 35pt 7pt 35pt} 
    ]
    {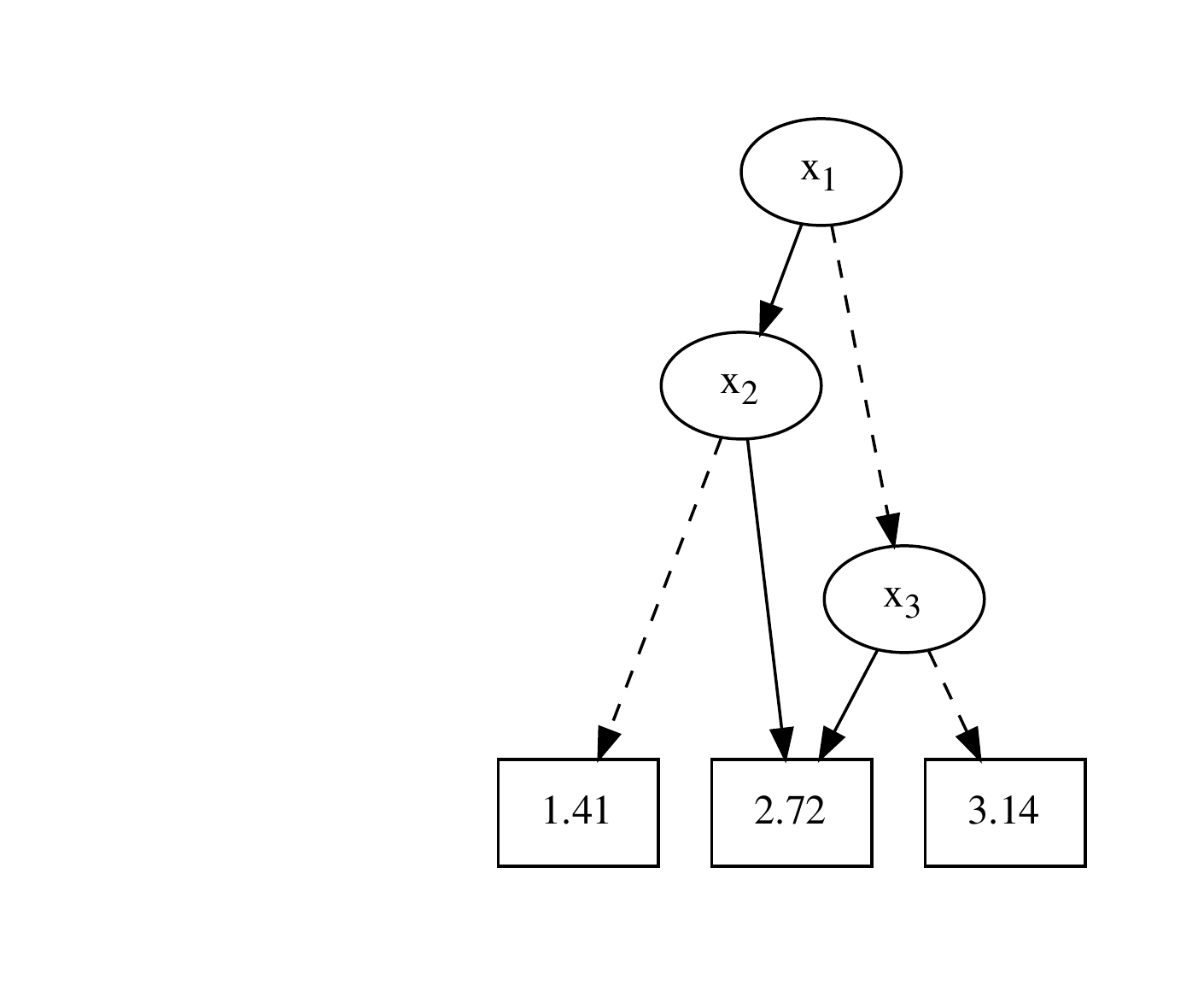}
    \caption{
        An ADD representing a PB function.
        If an edge from an oval node is solid (respectively dashed), then the corresponding variable is assigned $1$ (respectively $0$).
    }
\label{figAdd}
\end{figure}


\section{Evaluation}

We conducted computational experiments to answer the following \emph{empirical questions}.
\begin{enumerate}
    \item Does \dpo{} significantly contribute to a portfolio of state-of-the-art exact solvers on application benchmarks that are encoded as CNF formulas?
    \item Is there a class of \xcnf{} benchmarks on which \dpo{} outperforms existing tools?
\end{enumerate}

We used a high-performance computing cluster.
Each solver-benchmark pair was exclusively run on a single core of an Intel Xeon CPU (E5-2650 v2 at 2.60GHz) with a RAM cap of 25 GB and a time cap of 1000 seconds.

Source code, benchmarks, and experimental data are available in a public repository:
\begin{center}
    \url{https://github.com/vuphan314/DPO}
\end{center}

\subsection{Solvers}

We are aware of only one native \bmpe{} solver \cite{sang2007dynamic}, but its code is no longer available, according to the authors.
Fortunately, \bmpe{} can be reduced to weighted partial \ms{} as described in \cref{secRelatedWork}.
The only existing \xcnf{} \ms{} solver we know is \gauss{} \cite{soos2021gaussian}.

We also considered the top three solvers in the complete weighted track of the \ms{} Evaluation 2021: \cash{} \cite{lei2021cashwmaxsat}, \maxhs{} \cite{davies2011solving}, and \uwr{} \cite{piotrow2020uwrmaxsat}.
But we had to exclude \cash{} because it did not report numeric optimal costs; we contacted the authors and are waiting for their response.
So we compared three \ms{} solvers, \maxhs, \uwr, and \gauss, to our \bmpe{} tool, \dpo.
Since \maxhs{} and \uwr{} only work on pure CNF, we employed the Tseitin transformation on benchmarks in \xcnf.
\ms{} instances were created before \ms{} solvers were run, so the \ms{} reduction time and CNF encoding time were excluded from the total solving time.

\subsection{Benchmarks}

We used \cms{} \cite{soos2009extending} to guarantee that every instance is satisfiable.
The first benchmark suite comprises \bayescnfs{} literal-weighted CNF instances that were derived from Bayesian networks \cite{sang2005performing}.
The second suite was generated by us.
Adapting a recent study on \ms{} \cite{kyrillidis2022dpms}, we created random \emph{chain formulas} in \xcnf{} that have low-width project-join trees.
For given integers $n$ and $k$, a chain formula is a conjunction of $n - k + 1$ clauses.
Clause $i$ involves $k$ variables: $x_i, x_{i + 1}, \ldots, x_{i + k -1}$.
Each clause is randomly an XOR or a disjunction of literals.
The polarity of each literal is also uniformly randomized.
Each variable $x$ has random weights: $\wf_x\of{\va{x}{0}} = 10$ and $\wf_x\of{\va{x}{1}} = 100$, or vice versa.
For such a formula, there is a simple (left-deep) project-join tree with width $k$.
We generated \chaincnfs{} chain benchmarks with $n = 100, 110, \ldots, 300$ and $k = 10, 11, \ldots, 30$.
Recall that \dpo{} and \gauss{} natively handle \xcnf.
Since neither \maxhs{} nor \uwr{} accepts XOR, we employed the Tseitin transformation to obtain pure-CNF benchmarks, using \tool{PyEDA} \cite{drake2015pyeda}.

\subsection{Performance}

On the Bayesian benchmark suite, \maxhs{} performed very well, solving all \bayescnfs{} instances.
\dpo{} only solved 1014 and was faster than \maxhs{} on three of these benchmarks.
There was no need to run \uwr{} or \gauss.
The answer to the first empirical question is clear: on these application benchmarks, \dpo{} did not significantly contribute to the state of the art.

On the random chain formulas in \xcnf, \dpo{} outperformed \maxhs, \uwr, and \gauss.
An advantage \gauss{} and \dpo{} had was being able to directly process XOR, while \maxhs{} and \uwr{} had to solve larger CNF formulas after the Tseitin transformation.
See \cref{tableChain,figChain} for more detail on the chain benchmarks.

\begin{table}[H]
    \centering
    \caption{
        On most of the \chaincnfs{} chain benchmarks, \dpo{} was the fastest solver.
        \vbs0, a \emph{virtual best solver}, simulates running \maxhs, \uwr, and \gauss{} concurrently and finishing once one of these three actual solvers succeeds.
        \vbs1 simulates the portfolio of \maxhs, \uwr, \gauss, and \dpo.
    }
    \begin{tabular}{|l|r|r|r|r|r|} \hline
        \multirow{2}{*}{Solver} & \multirow{2}{*}{Mean peak RAM (GB)} & \multicolumn{3}{c|}{Benchmarks solved (of \chaincnfs)} & \multirow{2}{*}{Mean PAR-2 score} \\ \cline{3-5}
        & & Unique & Fastest & Total & \\ \hline
        \maxhs  & 0.10  &  0    &              0    & \chaincnfs    &  82.3 \\ \hline
        \uwr    & 0.07  &  0    &             12    &        329    & 624.8 \\ \hline
        \gauss  & 0.02  &  0    &             58    & \chaincnfs    &   6.1 \\ \hline
        \dpo    & 0.03  &  0    & \fastchaincnfs    & \chaincnfs    &   1.0 \\ \hline
        \vbs0   &   NA  & NA    &             NA    & \chaincnfs    &   5.6 \\ \hline
        \vbs1   &   NA  & NA    &             NA    & \chaincnfs    &   0.3 \\ \hline
    \end{tabular}
\label{tableChain}
\end{table}
\begin{figure}[H]
    \centering
    \input{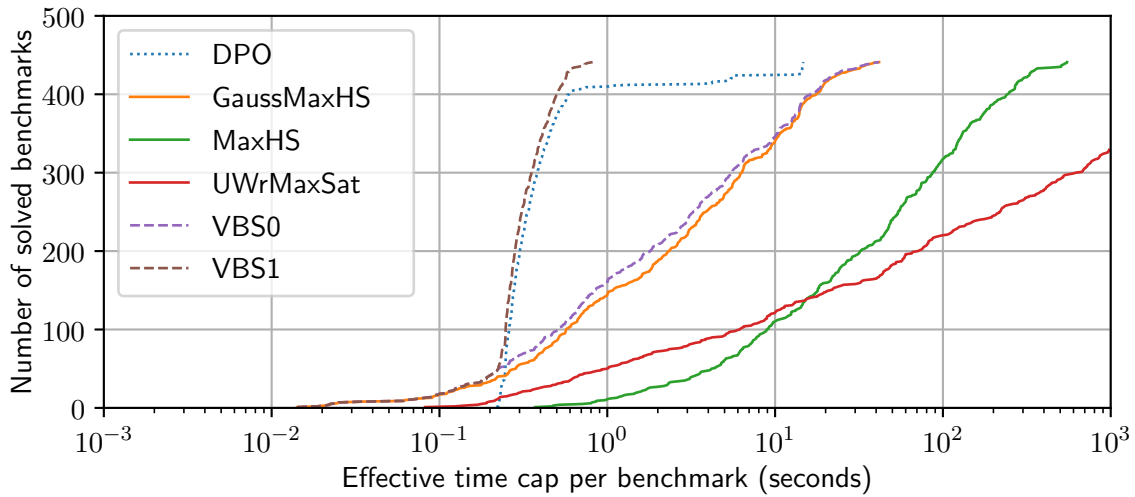}
    \caption{
        On this plot, each $\tup{x, y}$ point on a solver curve means that the solver would solve $y$ benchmarks in total if the time cap was $x$ seconds per benchmark.
        \dpo{} was the fastest solver on \fastchaincnfs{} of \chaincnfs{} chain benchmarks.
        The improvement of the virtual best solver \vbs1 over \vbs0 is due to \dpo.
    }
\label{figChain}
\end{figure}

Clearly, solvers that only accept pure-CNF benchmarks suffer when there are many hybrid constraints, even with the Tseitin transformation.
To answer the second empirical question, we identified a class of \xcnf{} chain formulas on which \dpo{} outperformed state-of-the-art \ms{} solvers.
This class includes hybrid Boolean formulas that have low-width project-join trees.


\section{Conclusion}

We introduce \dpo, a dynamic-programming optimizer that exactly solves \bmpe.
\dpo{} leverages techniques to build and execute project-join trees from a WMC solver \cite{dudek2020dpmc}.
\dpo{} also adapts an iterative procedure to find maximizers from \ms{} \cite{kyrillidis2022dpms}.
Our experiments show that \dpo{} can outperform state-of-the-art \ms{} solvers (\maxhs{} \cite{davies2011solving}, \uwr{} \cite{piotrow2020uwrmaxsat}, and \gauss{} \cite{soos2021gaussian}) by handling \xcnf{} natively and exploiting low-width project-join trees.

For future work, we plan to add support for hybrid inputs, such as PB and cardinality constraints.
Also, \dpo{} can be extended to solve more general problems, \eg, \emph{existential-random stochastic satisfiability} \cite{lee2018solving}, \emph{maximum model counting} \cite{fremont2017maximum}, and \emph{functional aggregate queries} \cite{abo2016faq}.
Another research direction is to improve \dpo{} with parallelism, as a portfolio solver (\eg, \cite{xu2008satzilla}) or with a multi-core ADD package (\eg, \cite{van2015sylvan}).


\appendix
\section{Proofs}


\subsection{Proof of \cref{propIterMax}}

\rePropIterMax*
\begin{proof}
    By \cref{defDsgn} (derivative sign), we have $\pars{\dsgn_x f}(\ta) = \set{\va{x}{1}}$ if $f(\extend{\ta}{x}{1}) \ge f(\extend{\ta}{x}{0})$, and $\pars{\dsgn_x f}(\ta) = \set{\va{x}{0}}$ otherwise.
    First, assume the former case, $\pars{\dsgn_x f}(\ta) = \set{\va{x}{1}}$.
    Then:
    \begin{align*}
        f\of{\ta \cup \pars{\dsgn_x f}(\ta)}
        & = f\of{\extend{\ta}{x}{1}}
        \\ & = \max\of{f\of{\extend{\ta}{x}{0}}, f\of{\extend{\ta}{x}{1}}} \tag*{as we assumed the former case}
        \\ & = (\exists_x f)(\ta) \tag*{by \cref{defExistProj} (existential projection)}
        \\ & = \pars{\exists_{X \setminus \set{x}}(\exists_x f)}(\emptyset) \tag*{since $\ta$ is a maximizer of $\exists_x f : \ps{X \setminus \set{x}} \to \R$}
        \\ & = \pars{\exists_X f}(\emptyset) \tag*{because existential projection is commutative}
    \end{align*}
    So $\ta \cup \pars{\dsgn_x f}(\ta)$ is a maximizer of $f : \ps{X} \to \R$ by \cref{defMaximizer} (maximizer).
    The latter case, $\pars{\dsgn_x f}(\ta) = \set{\va{x}{0}}$, is similar.
\end{proof}


\subsection{Proof of \cref{propMonoAlgo}}

\rePropMono*
\begin{proof}
    We prove that the two outputs of the algorithm are the maximum and a maximizer of $f_n$, as requested by the \bmpe{} problem (\cref{defBmpe}).

    Regarding the {first output} ($m$), on \cref{lineMonoMaximum} of \cref{monoAlgo}, note that $f_0 = \exists_{x_1} \exists_{x_2} \ldots \exists_{x_{n - 1}} \exists_{x_n} f_n = \exists_{X_n} f_n$.
    Then $m = f_0(\emptyset) = (\exists_{X_n} f_n)(\emptyset)$, which is the maximum of $f_n$ by \cref{defMaximum}.

    Regarding the {second output} ($\ta_n$), Lines \ref{lineMonoTa0}-\ref{lineMonoTaN} of the algorithm iteratively compute each truth assignment $\ta_i$ that is a maximizer of the PB function $f_i$.
    This process of iterative maximization is correct due to \cref{propIterMax}.
    Finally, $\ta_n$ is a maximizer of $f_n$.
\end{proof}


\subsection{Proof of \cref{propEarlyProj}}

\rePropEarlyProj*
\begin{proof}
    Let $x \in S$ be a variable.
    We first show that $\Sigma_x \pars{f \cdot g} = \pars{\Sigma_x f} \cdot g : \ps{X \cup Y \setminus \set{x}} \to \R$.
    For each truth assignment $\ta : X \cup Y \setminus \set{x} \to \B$, we have:
    \begin{align*}
        & \pars{\Sigma_x \pars{f \cdot g}}(\ta)
        \\ & = \pars{f \cdot g}\of{\extend{\ta}{x}{1}} + \pars{f \cdot g}\of{\extend{\ta}{x}{0}} \tag*{by \cref{defAddProj} (additive projection)}
        \\ & = f\of{\restrict{\extend{\ta}{x}{1}}{X}} \cdot g\of{\restrict{\extend{\ta}{x}{1}}{Y}} + f\of{\restrict{\extend{\ta}{x}{0}}{X}} \cdot g\of{\restrict{\extend{\ta}{x}{0}}{Y}} \tag*{by \cref{defJoin} (join)}
        \\ & = f\of{\restrict{\extend{\ta}{x}{1}}{X}} \cdot g\of{\restrict{\ta}{Y}} + f\of{\restrict{\extend{\ta}{x}{0}}{X}} \cdot g\of{\restrict{\ta}{Y}} \tag*{since $x \notin Y$}
        \\ & = \pars{f\of{\restrict{\extend{\ta}{x}{1}}{X}} + f\of{\restrict{\extend{\ta}{x}{0}}{X}}} \cdot g\of{\restrict{\ta}{Y}} \tag*{by common factor}
        \\ & = \pars{f\of{\extend{\restrict{\ta}{X}}{x}{1}} + f\of{\extend{\restrict{\ta}{X}}{x}{0}}} \cdot g\of{\restrict{\ta}{Y}} \tag*{as $x \in X$}
        \\ & = \pars{\Sigma_x f}\of{\restrict{\ta}{X}} \cdot g\of{\restrict{\ta}{Y}} \tag*{by definition of additive projection}
        \\ & = \pars{\Sigma_x f}\of{\restrict{\ta}{X \setminus \set{x}}} \cdot g\of{\restrict{\ta}{Y}} \tag*{because $x \notin \dom{\ta} = X \cup Y \setminus \set{x}$}
        \\ & = \pars{\Sigma_x f \cdot g}(\ta) \tag*{by definition of join}
    \end{align*}
    Since additive projection is commutative, we can generalize this equality from a single variable $x \in S$ to an equality on a whole set $S$ of variables: $\Sigma_S (f \cdot g) = (\Sigma_S f) \cdot g$.
    The case of existential projection, $\exists_S (f \cdot g) = (\exists_S f) \cdot g$, is similar.
\end{proof}


\subsection{Proof of \cref{thmRootValuation}}

\reThmRootVal*
\begin{proof}
    This theorem concerns \bmpe.
    A very similar theorem concerns WMC \cite[Theorem 2]{dudek2020dpmc}.
    We simply adapt that proof \cite[Section C.2]{dudek2020dpmc} and replace additive projection with existential projection.
\end{proof}


\subsection{Proof of \cref{thmDpAlgo}}
\label{secProofDp}

\cref{thmDpAlgo} concerns \cref{dpAlgo,valuatorAlgo}.
We actually prove the correctness of their annotated versions, which are respectively \cref{dpAlgoA,valuatorAlgoA}.

\newcommand{\actives}{A} 
\newcommand{\elims}{E} 
\newcommand{\childval}{h}

To simplify notations, for a multiset $\actives$ of PB functions $a$, define $\pb{\actives} := \prod_{a \in \actives} a$ and $\vars{\actives} := \bigcup_{a \in \actives} \vars{a}$.

\clearpage

\begin{algorithm}[H]
\caption{Dynamic Programming for \bmpe}
\label{dpAlgoA}
    \KwIn{$\phi$: an \xcnf{} formula}
    \KwIn{$\wf$: a literal-weight function over $X = \vars{\phi}$}
    \KwOut{$m \in \R$: the maximum of $\pb{\phi} \cdot \wf$}
    \KwOut{$\ta \in \ps{X}$: a maximizer of $\pb{\phi} \cdot \wf$}

    \DontPrintSemicolon
    $\T = \tup{T, r, \gamma, \pi} \gets$ a project-join tree for $\phi$\;
    $\stack \gets \tup{}$ \tcp{an initially empty stack}
    $\elims \gets \emptyset$ \tcp{$\elims \subseteq X$ is a set of variables that have been eliminated via projection} \label{lineInitProjVarsA}
    $\actives \gets \set{\pb{c} \mid c \in \phi} \cup \set{\wf_x \mid x \in X}$ \tcp{$\actives$ is a multiset of ``active'' PB functions} \label{lineInitActivesA}
    $\val{r} \gets \valuator(\phi, \T, \wf, r, \stack, \elims, \actives)$ \tcp{$\valuator$ (\cref{valuatorAlgoA}) modifies $\stack$, $\elims$, and $\actives$} \label{lineValRootA}
    $m \gets \val{r}(\emptyset)$ \tcp{$\val{r}$ is a constant PB function}
    $\ta \gets \emptyset$ \tcp{an initially empty truth assignment}
    \Assert{$\ta$ is a maximizer of $\exists_\elims \pars{\pb{\phi} \cdot \wf}$} \tcp{$\elims = X$ after \cref{lineValRootA}} \label{lineMaximizerConstA}
    \While{$\stack$ is not empty}{
        $\gx \gets \pop{\stack}$ \tcp{$\gx = \dsgn_x f$ is a derivative sign, where $x \in \elims$ is an unassigned variable and $f$ is a PB function}
        $\ta' \gets \ta$\;
        $\ta \gets \ta' \cup \gx\of{\restrict{\ta'}{\dom{\gx}}}$ \tcp{$\gx\of{\restrict{\ta'}{\dom{\gx}}} = \set{\va{x}{b}}$, where $b \in \B$}
        $\elims \gets \elims \setminus \set{x}$ \tcp{$x$ has just been assigned $b$ in $\ta$}
        \Assert{$\ta$ is a maximizer of $\exists_\elims \pars{\pb{\phi} \cdot \wf}$}\; \label{lineMaximizerPopA}
    }
    \Return $\tup{m, \ta}$ \tcp{$\elims = \emptyset$ now (all variables have been assigned in $\ta$)}
\end{algorithm}

\clearpage

\newcommand{\child}{u}

\begin{algorithm}[H]
\caption{$\valuator(\phi, \T, \wf, v, \stack, \elims, \actives)$}
\label{valuatorAlgoA}
    \KwIn{$\phi$: an \xcnf{} formula}
    \KwIn{$\T = \tup{T, r, \gamma, \pi}$: a project-join tree for $\phi$}
    \KwIn{$\wf$: a literal-weight function over $\vars{\phi}$}
    \KwIn{$v \in \V{T}$: a node}
    \KwIn{$\stack$: a stack (of derivative signs) that will be modified}
    \KwIn{$\elims$: a set (of variables) that will be modified}
    \KwIn{$\actives$: a multiset (of PB functions) that will be modified}
    \KwOut{$\val{v}$: the $\wf$-valuation of $v$}

    \DontPrintSemicolon
    \Assert{$\pb{\actives} = \exists_\elims \pars{\pb{\phi} \cdot \wf}$} \tcp{pre-condition} \label{linePrecondA}
    \If{$v \in \Lv$}{
        $f \gets \pb{\gamma(v)}$ \tcp{the Boolean function represented by the clause $\gamma(v) \in \phi$}
    }
    \Else(\tcp*[h]{$v$ is an internal node of $\tup{T, r}$}){
        $f \gets \set{\tup{\emptyset, 1}}$ \tcp{the PB function for multiplicative identity}
        $\minsert{\actives}{f}$\; \label{lineInsertIdA}
        \For(\tcp*[h]{non-empty set of nodes}){$\child \in \C{v}$}{
            $\childval \gets \valuator(\phi, \T, \wf, \child, \stack, \elims, \actives)$ \tcp{recursive call on a child $\child$ of $v$} \label{lineValChildA}
            $f' \gets f$\;
            $f \gets f' \cdot \childval$\;
            $\mremove{\actives}{\childval}$ \tcp{$\childval$ was added by initialization  \cref{lineInitActivesA} of \cref{dpAlgoA}) or a recursive call (\cref{lineInsertJoinA} or \ref{lineInsertProjA})}
            $\mremove{\actives}{f'}$ \tcp{$f'$ was added before this loop (\cref{lineInsertIdA}) or in the previous iteration (\cref{lineInsertJoinA})}
            $\minsert{\actives}{f}$\; \label{lineInsertJoinA}
        }
        \Assert{$\pb{\actives} = \exists_\elims \pars{\pb{\phi} \cdot \wf}$} \tcp{join-condition} \label{lineJoinCondA}
        \For(\tcp*[h]{possibly empty set of variables}){$x \in \pi(v)$}{
            $f' \gets f$\;
            $\gx \gets \dsgn_x f$\;
            $\push{\stack}{\gx}$ \tcp{$\stack$ will be used to construct a maximizer (\cref{dpAlgoA})}
            \Assert{$\ta \cup \gx\of{\restrict{\ta}{\dom{\gx}}}$ is a maximizer of $\exists_\elims \pars{\pb{\phi} \cdot \wf}$ if $\ta$ is a maximizer of $\exists_{\set{x} \cup \elims} \pars{\pb{\phi} \cdot \wf}$}\; \label{lineMaximizerPushA}
            $f \gets \exists_x \pars{f' \cdot \wf_x}$\;
            $\elims \gets \elims \cup \set{x}$\;
            $\mremove{\actives}{f', \wf_x}$ \tcp{$f'$ was added before this loop (\cref{lineInsertJoinA}) or in the previous iteration (\cref{lineInsertProjA}); $\wf_x$ was added by initialization  \cref{lineInitActivesA} of \cref{dpAlgoA})}
            $\minsert{\actives}{f}$\; \label{lineInsertProjA}
            \Assert{$\pb{\actives} = \exists_\elims \pars{\pb{\phi} \cdot \wf}$} \tcp{project-condition} \label{lineProjCondA}
        }
    }
    \Assert{$\pb{\actives} = \exists_\elims \pars{\pb{\phi} \cdot \wf}$} \tcp{post-condition} \label{linePostCondA}
    \Return $f$
\end{algorithm}

\begin{restatable}[Correctness of \cref{dpAlgoA}]{theorem}{reAnnThmDper}
\label{thmDpAlgoA}
    Let $\phi$ be an \xcnf{} formula and $\wf$ be a literal-weight function over $\vars{\phi}$.
    Then \cref{dpAlgoA} solves \bmpe{} on $\tup{\phi, \wf}$.
\end{restatable}

We will first prove some lemmas regarding \cref{valuatorAlgoA}.

\subsubsection{Proofs for \cref{valuatorAlgoA}}

The following proofs involve an \xcnf{} formula $\phi$, a project-join tree $\T = \tup{T, r, \gamma, \pi}$ for $\phi$, and a literal-weight function $\wf$ over $X = \vars{\phi}$.

\begin{lemma}
\label{lemmaPreCondRootA}
    The pre-condition of the call $\valuator(\phi, \T, \wf, r, \stack, \elims, \actives)$ holds (\cref{linePrecondA} of \cref{valuatorAlgoA}).
\end{lemma}
\begin{proof}
    \cref{valuatorAlgoA} is called for the first time on \cref{lineValRootA} of \cref{dpAlgoA}.
    We have:
    \begin{align*}
        \pb{\actives}
        & = \prod_{a \in \actives} a \tag*{by definition}
        \\ & = \prod_{c \in \phi} \pb{c} \cdot \prod_{x \in X} \wf_x \tag*{as initialized on \cref{lineInitActivesA} of \cref{dpAlgoA}}
        \\ & = \prod_{c \in \phi} \pb{c} \cdot \wf \tag*{since $\wf$ is a literal-weight function}
        \\ & = \pb{\phi} \cdot \wf \tag*{because $\phi$ is an \xcnf{} formula}
        \\ & = \exists_\emptyset \pars{\pb{\phi} \cdot \wf} \tag*{by convention}
        \\ & = \exists_\elims \pars{\pb{\phi} \cdot \wf} \tag*{as initialized on \cref{lineInitProjVarsA} of \cref{dpAlgoA}}
    \end{align*}
\end{proof}

To simplify proofs, for each internal node $v$ of $\T$, assume that the sets $\C{v}$ and $\pi(v)$ have arbitrary (but fixed) orders.
Then we can refer to members of these two sets as the first, second, \ldots, and last.

\begin{lemma}
\label{lemmaPreCondFirstChildA}
    Let $v$ be an internal node of $\T$ and $\child$ be the first node in the set $\C{v}$.
    Note that $\valuator(\phi, \T, \wf, v, \stack, \elims, \actives)$ calls $\valuator(\phi, \T, \wf, \child, \stack, \elims, \actives)$.
    Assume that the pre-condition of the caller ($v$) holds.
    Then the pre-condition of the callee ($\child$) holds.
\end{lemma}
\begin{proof}
    Before \cref{lineValChildA} of \cref{valuatorAlgoA}, $\actives$ is modified trivially: inserting the multiplicative identity $f$ does not change $\pb{\actives}$.
    Also, $\elims$ is not modified.
\end{proof}

\begin{lemma}
\label{lemmaPostCondLeafA}
    Let $v$ be a leaf of $\T$.
    Assume that the pre-condition of the call $\valuator(\phi, \T, \wf, v, \stack, \elims, \actives)$ holds.
    Then the post-condition holds (\cref{linePostCondA}).
\end{lemma}
\begin{proof}
    Neither $\actives$ nor $\elims$ is modified in the non-recursive branch of \cref{valuatorAlgoA}, \ie, when $v \in \Lv$.
\end{proof}

\begin{lemma}
\label{lemmaPreCondNextChildA}
    Let $v$ be an internal node of $\T$ and $\child_1 < \child_2$ be consecutive nodes in $\C{v}$.
    Note that $\valuator(\phi, \T, \wf, v, \stack, \elims, \actives)$ calls $\valuator(\phi, \T, \wf, \child_1, \stack, \elims, \actives)$ then calls $\valuator(\phi, \T, \wf, \child_2, \stack, \elims, \actives)$.
    Assume that the post-condition of the first callee ($\child_1$) holds.
    Then the pre-condition of the second callee ($\child_2$) holds.
\end{lemma}
\begin{proof}
    After the first callee returns and before the second callee starts, $\actives$ is trivially modified: replacing $\childval$ and $f'$ with $f = f' \cdot \childval$ does not change $\pb{\actives}$.
    Also, $\elims$ is not modified.
\end{proof}

\begin{lemma}
\label{lemmaJoinCondA}
    Let $v$ be an internal node of $\T$.
    For each $\child \in \C{v}$, note that $\valuator(\phi, \T, \wf, v, \stack, \elims, \actives)$ calls $\valuator(\phi, \T, \wf, \child, \stack, \elims, \actives)$.
    Assume that the post-condition of the last callee holds.
    Then the join-condition of the caller holds (\cref{lineJoinCondA}).
\end{lemma}
\begin{proof}
    After the last callee returns, $\actives$ is trivially modified: replacing $\childval$ and $f'$ with $f = f' \cdot \childval$ does not change $\pb{\actives}$.
    Also, $\elims$ is not modified.
\end{proof}

\begin{lemma}
\label{lemmaProjCondFirstVarA}
    Let $v$ be an internal node of $\T$ and $x$ be the first variable in $\pi(v)$.
    Assume that the join-condition of \cref{valuatorAlgoA} holds.
    Then the project-condition holds (\cref{lineProjCondA}).
\end{lemma}
\begin{proof}
    We have:
    \begin{align*}
        \pb{\actives}
        & = \pb{\set{f} \cup \actives \setminus \set{f}}
        \\ & = f \cdot \pb{\actives \setminus \set{f}}
        \\ & = \pars{\exists_x \pars{f' \cdot \wf_x}} \cdot \pb{\actives \setminus \set{f}}
        \\ & = \exists_x \pars{f' \cdot \wf_x \cdot \pb{\actives \setminus \set{f}}} \tag*{since $x \notin \vars{\actives \setminus \set{f}}$, as $\T$ is a project-join tree and $\wf$ is a literal-weight function}
        \\ & = \exists_x \pb{\set{f', \wf_x} \cup \actives \setminus \set{f}}
        \\ & = \exists_x \exists_{\elims \setminus \set{x}} \pars{\pb{\phi} \cdot \wf} \tag*{because the join-condition was assumed to hold}
        \\ & = \exists_\elims \pars{\pb{\phi} \cdot \wf}
    \end{align*}
\end{proof}

\begin{lemma}
\label{lemmaProjCondNextVarA}
    Let $v$ be an internal node of $\T$ and $x_1 < x_2$ be consecutive variables in $\pi(v)$.
    Assume that the project-condition of \cref{valuatorAlgoA} holds in the iteration for $x_1$.
    Then the project-condition also holds in the iteration for $x_2$.
\end{lemma}
\begin{proof}
    Similar to \cref{lemmaProjCondFirstVarA}.
\end{proof}

\begin{lemma}
\label{lemmaProjCondA}
    Let $v$ be an internal node of $\T$ and $x$ be a variable in $\pi(v)$.
    Assume that the join-condition of \cref{valuatorAlgoA} holds.
    Then the project-condition holds.
\end{lemma}
\begin{proof}
    Apply \cref{lemmaProjCondFirstVarA,lemmaProjCondNextVarA}.
\end{proof}

\begin{lemma}
\label{lemmaPostCondInternalA}
    Let $v$ be an internal node of $\T$.
    For each $\child \in \C{v}$, note that $\valuator(\phi, \T, \wf, v, \stack, \elims, \actives)$ calls $\valuator(\phi, \T, \wf, \child, \stack, \elims, \actives)$.
    Assume that the post-condition of the last callee holds.
    Then the post-condition of the caller holds.
\end{lemma}
\begin{proof}
    By \cref{lemmaJoinCondA}, the join-condition holds.
    Then by \cref{lemmaProjCondA}, the project-condition holds in the iteration for the last variable in $\pi(v)$.
    So the post-condition holds.
\end{proof}

\begin{lemma}
\label{lemmaPrePostCondsA}
    The pre-condition and post-condition of \cref{valuatorAlgoA} hold for each input $v \in \V{T}$.
\end{lemma}
\begin{proof}
    Employ induction on the recursive calls to the algorithm.

    In the \textbf{base case}, consider the leftmost nodes of $\T$: $v_1, v_2, \ldots, v_k$, where $v_{i + 1}$ is the first node in $\C{v_i}$, $v_1$ is the root $r$, and $v_k$ is a leaf.
    We show that the pre-condition of the call $\valuator(\phi, \T, \wf, v_i, \stack, \elims, \actives)$ holds for each $i = 1, 2, \ldots, k$ and that the post-condition of the call $\valuator(\phi, \T, \wf, v_k, \stack, \elims, \actives)$ holds.
    The pre-condition holds for each $v_i$ due to \cref{lemmaPreCondRootA,lemmaPreCondFirstChildA}.
    The post-condition holds for $v_k$ by \cref{lemmaPostCondLeafA}.

    In the \textbf{step case}, the induction hypothesis is that the pre-condition holds for each call started before the current call starts and that the post-condition holds for each call returned before the current call starts.
    We show that the pre-condition and post-condition hold for the current call $\valuator(\phi, \T, \wf, v, \stack, \elims, \actives)$.
    The pre-condition for $v$ holds due to \cref{lemmaPreCondFirstChildA,lemmaPreCondNextChildA}.
    The post-condition for $v$ holds by \cref{lemmaPostCondLeafA,lemmaPostCondInternalA}.
\end{proof}

\begin{lemma}
\label{lemmaJoinProjCondsA}
    The join-condition and project-condition of \cref{valuatorAlgoA} hold for each internal node $v$ of $\T$ and variable $x$ in $\pi(v)$.
\end{lemma}
\begin{proof}
    By \cref{lemmaPrePostCondsA}, the post-condition of the last call on \cref{lineValChildA} holds.
    Then by \cref{lemmaJoinCondA}, the join-condition holds.
    And by \cref{lemmaProjCondA}, the project-condition holds.
\end{proof}

\newcommand{\other}{g} 

\begin{lemma}
\label{lemmaDsgnEarlyProjA}
    Let $f : \ps{X} \to \R$ and $\other : \ps{Y} \to \R$ be PB functions with non-negative ranges.
    Assume that $x$ is a variable in $X \setminus Y$ and that $\ta$ is a truth assignment for $Y \cup X \setminus \set{x}$.
    Then ${\pars{\dsgn_x f}\of{\restrict{\ta}{X \setminus \set{x}}}} = \pars{\dsgn_x (f \cdot \other)}(\ta)$.
\end{lemma}
\begin{proof}
    First, assume the case $\pars{\dsgn_x f}\of{\restrict{\ta}{X \setminus \set{x}}} = \set{\va{x}{1}}$.
    We have the following inequalities:
    \begin{align*}
        f\of{\extend{\restrict{\ta}{X \setminus \set{x}}}{x}{1}} & \ge f\of{\extend{\restrict{\ta}{X \setminus \set{x}}}{x}{0}} \tag*{by \cref{defDsgn} (derivative sign)}
        \\ f\of{\restrict{\extend{\ta}{x}{1}}{X}} & \ge f\of{\restrict{\extend{\ta}{x}{0}}{X}} \tag*{as $x \in X$}
        \\ f\of{\restrict{\extend{\ta}{x}{1}}{X}} \cdot \other\of{\restrict{\ta}{Y}} & \ge f\of{\restrict{\extend{\ta}{x}{0}}{X}} \cdot \other\of{\restrict{\ta}{Y}} \tag*{because the range of $\other$ is non-negative}
        \\ f\of{\restrict{\extend{\ta}{x}{1}}{X}} \cdot \other\of{\restrict{\extend{\ta}{x}{1}}{Y}} & \ge f\of{\restrict{\extend{\ta}{x}{0}}{X}} \cdot \other\of{\restrict{\extend{\ta}{x}{0}}{Y}} \tag*{since $x \notin Y$}
        \\ (f \cdot \other)\of{\extend{\ta}{x}{1}} & \ge (f \cdot \other)\of{\extend{\ta}{x}{0}} \tag*{by \cref{defJoin} (join)}
    \end{align*}
    So $\pars{\dsgn_x (f \cdot \other)}(\ta) = \set{\va{x}{1}}$ by definition.
    Thus ${\pars{\dsgn_x f}\of{\restrict{\ta}{X \setminus \set{x}}}} = \pars{\dsgn_x (f \cdot \other)}(\ta)$.
    The case $\pars{\dsgn_x f}\of{\restrict{\ta}{X \setminus \set{x}}} = \set{\va{x}{0}}$ is similar.
\end{proof}

\begin{lemma}
\label{lemmaMaximizerDsgnA}
    Let $f : \ps{X} \to \R$ and $\other : \ps{Y} \to \R$ be PB functions with non-negative ranges.
    Assume that $x$ is a variable in $X \setminus Y$ and that $\ta$ is a maximizer of $\exists_x (f \cdot \other)$.
    Then $\ta' = \ta \cup \pars{\dsgn_x f}\of{\restrict{\ta}{X \setminus \set{x}}}$ is a maximizer of $f \cdot \other$.
\end{lemma}
\begin{proof}
    By \cref{lemmaDsgnEarlyProjA}, $\ta' = \ta \cup \pars{\dsgn_x (f \cdot \other)}(\ta)$.
    Then by \cref{propIterMax} (iterative maximization), $\ta'$ is a maximizer of $f \cdot \other$.
\end{proof}

\begin{lemma}
\label{lemmaMaximizerPushA}
    Let $v$ be an internal node of $\T$ and $x$ be a variable in $\pi(v)$.
    Then the assertion on \cref{lineMaximizerPushA} of \cref{valuatorAlgoA} holds.
\end{lemma}
\begin{proof}
    On that line:
    \begin{align*}
        \exists_{\set{x} \cup \elims} \pars{\pb{\phi} \cdot \wf}
        & = \exists_x \exists_\elims \pars{\pb{\phi} \cdot \wf}
        \\ & = \exists_x \pb{\actives} \tag*{as the join-condition and project-condition hold by \cref{lemmaJoinProjCondsA}}
        \\ & = \exists_x \pb{\set{f} \cup \actives \setminus \set{f}}
        \\ & = \exists_x \pars{f \cdot \pb{\actives \setminus \set{f}}}
    \end{align*}
    Apply \cref{lemmaMaximizerDsgnA}.
\end{proof}

\subsubsection{Proofs for \cref{dpAlgoA}}

\begin{lemma}
\label{lemmaMaximizerConstA}
    The assertion on \cref{lineMaximizerConstA} of \cref{dpAlgoA} holds.
\end{lemma}
\begin{proof}
    \cref{lineValRootA} changed $\elims$ to $X$ after all calls to $\valuator$ returned.
    Then $\exists_\elims \pars{\pb{\phi} \cdot \wf}$ is a constant PB function.
    So $\ta = \emptyset$ is the maximizer of $\exists_\elims \pars{\pb{\phi} \cdot \wf}$.
\end{proof}

\begin{lemma}
\label{lemmaMaximizerPopFirstDsgnA}
    Let $\gx$ be the first derivative sign to be popped in \cref{dpAlgoA}.
    Then the assertion on \cref{lineMaximizerPopA} holds at this time.
\end{lemma}
\begin{proof}
    Note that $\gx = \dsgn_x f$ for some variable $x \in X$ and PB function $f$.
    Consider the following on \cref{lineMaximizerPopA}.
    By \cref{lemmaMaximizerConstA}, the truth assignment $\ta'$ is a maximizer of $\exists_{\set{x} \cup \elims} \pars{\pb{\phi} \cdot \wf}$.
    Recall that when $\gx$ was pushed onto $\stack$, the assertion on \cref{lineMaximizerPushA} of \cref{valuatorAlgoA} holds (by \cref{lemmaMaximizerPushA}).
    Then $\ta$ is a maximizer of $\exists_\elims \pars{\pb{\phi} \cdot \wf}$.
\end{proof}

\begin{lemma}
\label{lemmaMaximizerPopNextDsgnA}
    Let \cref{dpAlgoA} successively pop derivative signs $\gx_1$ then $\gx_2$.
    Assume that the assertion on \cref{lineMaximizerPopA} holds in the iteration for $\gx_1$.
    Then the assertion also holds in the iteration for $\gx_2$.
\end{lemma}
\begin{proof}
    Note that $\gx_2 = \dsgn_{x_2} f_2$ for some variable $x_2 \in X$ and PB function $f_2$.
    Consider the following on \cref{lineMaximizerPopA} in the iteration for $\gx_2$.
    By the assumption, the truth assignment $\ta'$ is a maximizer of $\exists_{\set{x_2} \cup \elims} \pars{\pb{\phi} \cdot \wf}$.
    Recall that when $\gx_2$ was pushed onto $\stack$, the assertion on \cref{lineMaximizerPushA} of \cref{valuatorAlgoA} holds (by \cref{lemmaMaximizerPushA}).
    Then $\ta$ is a maximizer of $\exists_\elims \pars{\pb{\phi} \cdot \wf}$.
\end{proof}

We are now ready to prove \cref{thmDpAlgoA}.

\reAnnThmDper*
\begin{proof}
    The first output ($m$) is the maximum of $\pb{\phi} \cdot \wf$ as \cref{lineValRootA} computes the $\wf$-valuation of the root $r$ of a project-join tree for $\phi$ (see \cref{thmRootValuation} and \cref{lemmaValuator}).
    The second output ($\ta$) is a maximizer of $\pb{\phi} \cdot \wf$ by \cref{lemmaMaximizerPopFirstDsgnA,lemmaMaximizerPopNextDsgnA}.
\end{proof}


\bibliography{dpo.bib}


\end{document}